\documentclass[letterpaper,10pt,conference]{ieeeconf} % Comment this line out if you need a4paper

\IEEEoverridecommandlockouts % This command is only needed if you want to use the \thanks command

\usepackage{amsmath,amssymb,amsthm,mathtools}
\usepackage[dvipsnames]{xcolor}
\usepackage{graphicx}
\usepackage{booktabs}
\usepackage{array}
\usepackage[compress]{cite}
\usepackage[multiple]{footmisc}

\makeatletter
\newtheoremstyle{stylename}% name of the style to be used
  {\z@}% measure of space to leave above the theorem. E.g.: 3pt
  {\z@}% measure of space to leave below the theorem. E.g.: 3pt
  {}% name of font to use in the body of the theorem
  {\parindent}% measure of space to indent
  {\itshape}% name of head font
  {:}% punctuation between head and body
  {5\p@ plus\p@ minus\p@\relax}% space after theorem head; " " = normal interword space
  {\thmname{#1}\thmnumber{ #2}\thmnote{ (#3)}}
\def\QED{\mbox{\rule[0pt]{1.3ex}{1.3ex}}}

\newlength{\myparindent}
\makeatother
\theoremstyle{stylename}

\newtheorem{theorem}{Theorem}

\newtheorem{corollary}[theorem]{Corollary}
\newtheorem{lemma}[theorem]{Lemma}
\newtheorem{proposition}[theorem]{Proposition}

\newtheorem{definition}{Definition}

\newtheorem{problem}{Problem}

\newcommand{\myproofname}{Proof}

\renewcommand{\Re}{\mathbb{R}}
\newcommand{\Co}{\mathbb{C}}
\newcommand{\NNb}{\mathbb{N}}
\newcommand{\ZZb}{\mathbb{Z}}
\newcommand{\QQb}{\mathbb{Q}}
\newcommand{\nxn}{{n\times n}}
\newcommand{\mxn}{{m\times n}}

\newcommand{\diff}{\mathrm{d}}

\newcommand{\calC}{\mathcal{C}}

\newcommand{\calM}{\mathcal{M}}

\newcommand{\calR}{\mathcal{R}}

\newcommand{\bitsize}{\mathfrak{b}}
\newcommand{\vecc}{\mathrm{vec}}
\newcommand{\euler}{\mathrm{e}}
\newcommand{\pred}{\tilde{p}}
\newcommand{\pext}{p^{\text{ext}}}

\newcommand{\mext}{m^{\text{ext}}}
\newcommand{\Vred}{\tilde{V}}
\newcommand{\Vext}{V^{\text{ext}}}

\newcommand{\pdisc}{\hat{p}}
\newcommand{\remain}{\mathrm{rem}}
\newlength\algoindent
\newcommand{\new}{\color{blue}}

\newcounter{changes}

\newcommand{\extended}[1]{{\color{orange}#1}}
\newcommand{\short}[1]{{\color{ForestGreen}#1}}
\renewcommand{\extended}[1]{#1}
\renewcommand{\short}[1]{#1}
\renewcommand{\short}[1]{}

\title{\LARGE\bfseries Complexity of the LTI system trajectory boundedness problem}
\author{Guillaume O.\ Berger and Rapha\"el M.\ Jungers%
\thanks{GB is a FRIA (F.R.S.--FNRS) fellow.
RJ is a FNRS honorary Research Associate.
This project has received funding from the European Research Council under the European Union's Horizon 2020 research and innovation program under grant agreement No.\ 864017 -- L2C.
RJ is also supported by the Walloon Region, the Innoviris Foundation, and the FNRS (Chist-Era Druid-net).
Both are with ICTEAM institute, UCLouvain, Belgium.
Emails: \{guillaume.berger,raphael.jungers\}@uclouvain.be}}

\begin{document}

\maketitle
\thispagestyle{empty}
\pagestyle{empty}

%%%%%%%%%%%%%%%%%%%%%%%%%%%%%%%%%%%%%%%%%%%%%%%%%%%%%%%%%%%%%%%%%%%%%%%%%%%%%%%%
\begin{abstract}
% We show that the problems of deciding whether a continuous-time or discrete-time linear system with rational constant coefficients has bounded trajectories or not can be answered in polynomial time with respect to the bit size of the input coefficients.
% Our proofs combine several well-known results from numerical algebra for the computation of the determinant of matrices and for the location of the roots of a polynomial, and from control theory for the characterization of systems with bounded trajectories.
% By putting these results together and verifying that the pathological cases (e.g., in the location of the roots of a polynomial) do not occur, we come up with a polynomial-time algorithm for the considered problem.
% We also would like to mention the didactic flavor of this work: the different algorithms are fully described and the proofs of the intermediate complexity results are provided, so that the reader does not need to resort to external resource to understand the overall algorithmic procedure.
We study the algorithmic complexity of the problem of deciding whether a Linear Time Invariant dynamical system with rational coefficients has bounded trajectories.
Despite its ubiquitous and elementary nature in Systems and Control, it turns out that this question is quite intricate, and, to the best of our knowledge, unsolved in the literature.
We show that classical tools, such as Gaussian Elimination, the Routh--Hurwitz Criterion, and the Euclidean Algorithm for GCD of polynomials indeed allow for an algorithm that is polynomial in the bit size of the instance.
However, all these tools have to be implemented with care, and in a non-standard way, which relies on an advanced analysis.
\end{abstract}

%%%%%%%%%%%%%%%%%%%%%%%%%%%%%%%%%%%%%%%%%%%%%%%%%%%%%%%%%%%%%%%%%%%%%%%%%%%%%%%%
\section{Introduction}\label{sec-introduction}

This paper deals with the computational problem of deciding whether a Linear Time Invariant (LTI) dynamical system with rational coefficients has bounded trajectories; see Problems 1 and 2 in Section \ref{sec-problem-results}.
We show that this problem can be solved in polynomial time with respect to the bit size of the coefficients.
We are interested in the \emph{exact complexity}, also called ``bit complexity'' or ``complexity in the Turing model'', which accounts for the fact that arithmetic operations ($+$, $-$, $\times$, $\div$) on integers and rational numbers take a time proportional to the bit size of the operands.

Rational matrices appear in many applications, including combinatorics, computer science and information theory; for instance, the number of paths of length $r$ in a graph (involved for instance in the computation of its entropy \cite{lind1995anintroduction}) grows at most as $\rho^r$ ($\rho\geq0$) if and only if the adjacency matrix of the graph divided by $\rho$ has bounded powers.
However, despite its ubiquitous and paradigmatic nature for many applications, it seems that the question of the complexity of the problem of deciding whether the trajectories of a LTI system with rational coefficients are bounded is unsolved in the literature.
The aim of this paper is to fill this gap by providing a proof of its polynomial complexity.

% Therefore, because of its paradigmatic and elementary nature for many applications, it is important to understand the intrinsic complexity of the problem of deciding whether a LTI system with integer (or rational) coefficients has bounded trajectories.
% Therefore, it seems important to us to understand the intrinsic complexity of this problem.

\begin{figure}
\centering
\includegraphics[width=\linewidth,trim=12pt 5pt 16pt 8pt, clip]{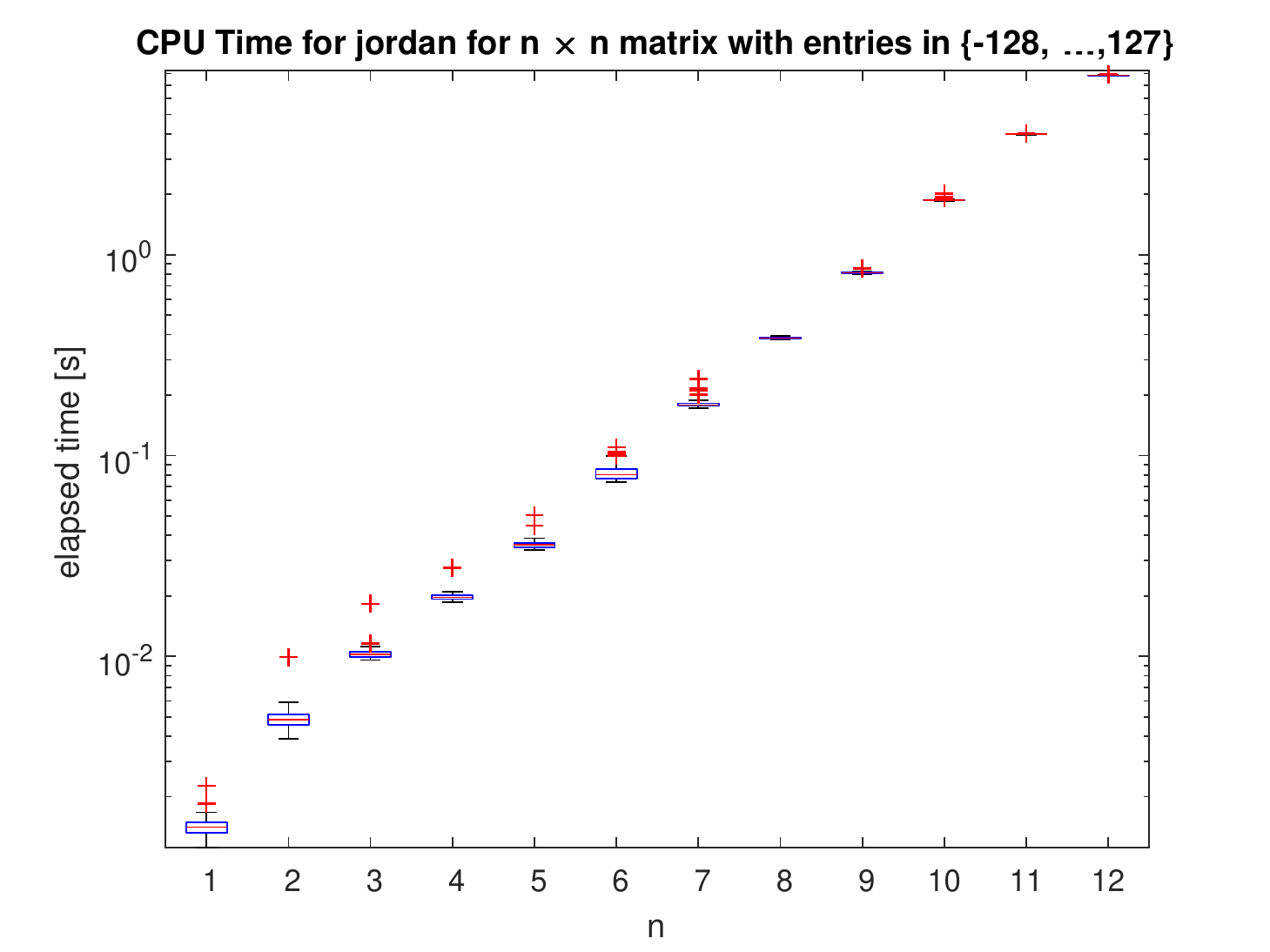}
\caption{A naive way to decide the boundedness of the trajectories of a LTI system described by a rational matrix would be to use the Matlab function \texttt{jordan} and look at the eigenvalues and size of the associated Jordan blocks.
However, this method will not be efficient for two reasons.
Firstly, the complexity of the Matlab function \texttt{jordan} applied on symbolic integer matrices seems \emph{super-linear}.
In the above plot, we have generated, for each $n\in\{1,\ldots,12\}$, $50$ random $n\times n$ matrices with $8$-bit integer entries (the bit size of the input is thus $8n^2$) and we measured the time to compute the Jordan form (represented by the boxplots).
The empirical complexity is clearly super-linear.
Secondly, the obtained Jordan form does not contain explicitly the eigenvalues of the matrix, but rather expressions of the form \texttt{root('some polynomial')}, so that extra computation is needed to decide the (marginal) stability of each of the Jordan blocks.
In this work, we show that a more efficient implementation is possible.}
\label{fig-jordan-matlab}
\end{figure}

The question of deciding \emph{asymptotic stability} (rather than boundedness of the trajectories) of LTI systems has received a lot of attention in the literature \cite{routh1877treatise,hurwitz1895ueber,lienard1914surlesigne}.
Algorithms with polynomial bit complexity have been proposed to address this question for systems with rational coefficients.
This includes algorithms based on the \emph{Routh--Hurwitz stability criterion} \cite{pena2004characterizations} where care is taken to avoid exponential blow-up of the bit size of the intermediate coefficients; or algorithms based on the resolution of the \emph{Lyapunov equation} \cite{bartels1972algorithm}, again taking care to avoid exponential blow-up of the bit size of the intermediate steps of the resolution.

However, these algorithms focus on the asymptotic stability and do not extend straightforwardly for the problem of deciding boundedness of the trajectories.
For instance, the Routh--Hurwitz stability criterion applied on the characteristic polynomial of a matrix allows to decide whether all eigenvalues of the matrix have negative real part.
For the problem of boundedness of the trajectories, the analysis is more difficult because eigenvalues with nonnegative real parts are also allowed%
\footnote{This holds for the continuous-time case, but a similar result holds for the discrete-time case.
Both cases are studied in this paper.}%
, provided they are on the imaginary axis and correspond to Jordan blocks of size one in the Jordan normal form of the system matrix.
Extensions of the Routh--Hurwitz criterion to compute the \emph{number} of roots with negative real part of a given polynomial have also been proposed in the literature (see, e.g., \cite[\S15]{gantmacher2000thetheory}, \cite{chen1993new,choghadi2013routh}), but the analysis of the bit complexity of these algorithms remains elusive.
The situation is similar for the approach based on the Lyapunov equation \cite{bartels1972algorithm}.
More precisely, while for the study of the asymptotic stability, a solution of the Lyapunov matrix inequality can be computed by arbitrarily fixing the right-hand side term to $-I$, this trick cannot be used for the analysis of the boundednes of the trajectories, since the RHS term is not guaranteed to be negative definite.
One has thus to solve a matrix \emph{inequality} instead of a matrix \emph{equation}, and there is to the best of the authors' knowledge no clear result available in the literature on the bit complexity of solving LMIs, so that an extension of the Lyapunov method for the problem of trajectory boundedness is not straightforward.

\emph{Objectives and methodology.}
The discussion above nevertheless suggests an algorithmic procedure for our problem, consisting in
\begin{enumerate}
\item computing the minimal polynomial of the system matrix, which contains the information on the eigenvalues and on the size of the largest associated Jordan block in the Jordan normal form of the system matrix;
\item using an extension of the Routh--Hurwitz criterion to decide whether all the roots of the minimal polynomial either have negative real part, or are on the imaginary axis and correspond to Jordan blocks of size one.
\end{enumerate}
We provide self-contained proofs that these two steps can be achieved in polynomial time with respect to the bit size of the matrix.
In particular, we provide a careful analysis of the extended Routh--Hurwitz criterion, showing that it provides a polynomial bit-complexity algorithm for the second step.%

\emph{Comparison with the literature.}
Algorithms with polynomial bit complexity for computing the minimal polynomial of a rational matrix have been proposed in the literature \cite{dumas2005efficient,dumas2006bounds}.%
\footnote{Note that the Matlab function \texttt{jordan} would not have helped us for this problem, as the complexity of this function seems to be super-linear in the bit size of the instance; see, e.g., Figure \ref{fig-jordan-matlab}.}
While being aware of these results, we describe here an \emph{elementary} algorithm for this problem, based on the defining property of the minimal polynomial (see Subsection \ref{ssec-minimal-polynomial}).
The proposed algorithm is likely to be less efficient than those available in the literature, but its elementary nature allows us to provide a simple self-contained proof of its polynomial bit complexity.

As for the second step, several extensions of the Routh--Hurwitz criterion have been proposed in the literature to compute for a given polynomial the number of roots on the imaginary axis and their multiplicity \cite{gantmacher2000thetheory,chen1993new,choghadi2013routh}.
However, no proofs of the polynomial complexity of these algorithms are provided.
In particular, since they are extensions of the classical Routh--Hurwitz criterion, there is no guarantee on the boundedness of the bit size of the intermediate coefficients; see, e.g., \cite[p.\ 321]{pena2004characterizations} for a discussion on the ``bit size growth factor'' for the Routh--Hurwitz criterion.
The extended Routh--Hurwitz criterion proposed in this paper draws on these results and combines them with techniques introduced in the context of the classical Routh--Hurwitz criterion to avoid ``bit-size blow-up''.
This results in a sound elementary algorithm to address the second step, and for which we provide a simple self-contained proof of the polynomial bit complexity.

\emph{Outline.}
The paper is organized as follows.
The statement of the problems and the main results are presented in Section \ref{sec-problem-results}.
Some preliminary results, namely on the computation of the determinant and the resolution of systems of linear equations, are presented in Section \ref{sec-preliminaries}.
Then, in Section \ref{sec-proof-thm-boundedness-cont}, we present the proof of the main result for continuous-time systems.
Finally, in Section \ref{sec-proof-thm-boundedness-disc}, we present the proof of the main result for discrete-time systems.

Many of the intermediate results used in our analysis are inspired from classical results, but have been adapted for the needs of this work.
Below, we refer to these results as \emph{folk theorems}, meaning that we are referring to a classical result --- possibly slightly adapted to our needs, but on which we do not claim any paternity.

% The discussion above suggests another straightforward method consisting in (i) computing the Jordan normal form of the system matrix, which can be done in exact arithmetic, and (ii) check whether the eigenvalues of the Jordan blocks satisfy the conditions mentioned above.
% However, this approach suffers from two limitations: first, computing the Jordan normal form is expensive, in particular experiments suggest that the Matlab implementation of $\mathtt{jordan}(A)$ is super-linear in the bit size of $A$ (see Figure~).

\emph{Notation.}
We use a Matlab-like notation for the indexing of submatrices; e.g., $A_{[1:r,:]}$ denotes the submatrix consisting of the $r$ first rows of the matrix $A$.
The degree of a polynomial $p$ is denoted by $\deg\,p$.
We use $i$ to denote the imaginary unit $i=\sqrt{-1}$ \emph{and} as an index $i\in\NNb$, but the disambiguation should be clear from the context.

%%%%%%%%%%%%%%%%%%%%%%%%%%%%%%%%%%%%%%%%%%%%%%%%%%%%%%%%%%%%%%%%%%%%%%%%%%%%%%%%
\section{Problem statement and main results}\label{sec-problem-results}

We start with the definition of bit size for integers, integer matrices and rational matrices.

\begin{definition}[Bit size]
\begin{itemize}
\item The \emph{bit size} of an integer $a\in\ZZb$ is $\bitsize(a)=\lceil\log_2(\lvert a\rvert+1)\rceil+1$ ($=$ smallest $b\in\ZZb_{\geq0}$ such that $-2^{b-1}+1\leq a\leq 2^{b-1}-1$).
\item The \emph{bit size} of an integer matrix $A=(a_{ij})_{i=1,j=1}^{m,n}\in\ZZb^\mxn$ is $\bitsize(A)=\sum_{i=1,j=1}^{m,n}\bitsize(a_{ij})$.
% We also let $\bitsize^*(A)=\max_{1\leq i\leq m,1\leq j\leq n}\bitsize(a_{ij})$.
\item The \emph{bit size} of a rational matrix $A\in\QQb^\mxn$, described by $A=B/q$ with $B\in\ZZb^\mxn$ and $q\in\ZZb_{>0}$, is $\bitsize(A)=\bitsize(B)+\bitsize(q)$.
\end{itemize}\vskip0pt
\end{definition}

We consider the following decision problem, accounting for the boundedness of the trajectories of continuous-time LTI systems:

\smallskip

\noindent\fbox{\parbox{\dimexpr\linewidth-2\fboxsep-2\fboxrule}{\setlength{\parindent}{\dimexpr\myparindent-\fboxsep-\fboxrule}%
\begin{problem}\label{prob-boundedness-cont}
Given a rational matrix $A\in\QQb^\nxn$, decide whether $\smash{\sup_{t\in\Re_{\geq0}}}\,\lVert e^{At}\rVert<\infty\rule[-4pt]{0pt}{0pt}$.
\end{problem}%
}}

\smallskip

% \begin{remark}
% Note that Problem \ref{prob-boundedness-cont} extends trivially to \emph{rational} matrices $\tilde{A}\in\QQb^\nxn$.
% Indeed, if $\tilde{A}=A/q$ with $A\in\ZZb^\nxn$ and $q\in\ZZb_{>0}$, then there is $C\geq0$ such that $\lVert e^{\tilde{A}t}\rVert\leq C$ for all $t\in\Re_{\geq0}$ if and only $A$ is a positive instance of Problem \ref{prob-boundedness-cont}.
% \end{remark}

The following theorem states that Problem \ref{prob-boundedness-cont} can be solved in polynomial time with respect to the bit size of the input.

\begin{theorem}\label{thm-boundedness-cont}
There is an algorithm that, given any $A\in\QQb^\nxn$, gives the correct answer to Problem \ref{prob-boundedness-cont}, and whose bit complexity is polynomial in $\bitsize(A)$.
\end{theorem}

The proof of Theorem \ref{thm-boundedness-cont} is presented in Section \ref{sec-proof-thm-boundedness-cont}.
Note that a rational matrix $A=B/q$, with $B\in\ZZb^\nxn$ and $q\in\Re_{>0}$, is a positive instance of Problem \ref{prob-boundedness-cont} if and only if $B$ is a positive instance of Problem \ref{prob-boundedness-cont}.
Hence, in Section \ref{sec-proof-thm-boundedness-cont}, we limit ourselves to proving Theorem \ref{thm-boundedness-cont} for integer matrices.

The same kind of results can be obtained for the problem of the boundedness of the trajectories of discrete-time LTI systems:

\smallskip

\noindent\fbox{\parbox{\dimexpr\linewidth-2\fboxsep-2\fboxrule}{\setlength{\parindent}{\dimexpr\myparindent-\fboxsep-\fboxrule}%
\begin{problem}\label{prob-boundedness-disc}
Given a rational matrix $A\in\QQb^\nxn$, decide whether $\smash{\sup_{t\in\ZZb_{\geq0}}}\,\lVert A^t\rVert<\infty\rule[-4pt]{0pt}{0pt}$.
\end{problem}%
}}

\smallskip

\begin{theorem}\label{thm-boundedness-disc}
There is an algorithm that, given any $A\in\QQb^\nxn$, gives the correct answer to Problem \ref{prob-boundedness-disc}, and whose bit complexity is polynomial in $\bitsize(A)$.
\end{theorem}

The proof of Theorem \ref{thm-boundedness-disc} is presented in Section \ref{sec-proof-thm-boundedness-disc}.

%%%%%%%%%%%%%%%%%%%%%%%%%%%%%%%%%%%%%%%%%%%%%%%%%%%%%%%%%%%%%%%%%%%%%%%%%%%%%%%%
\section{Preliminary results}\label{sec-preliminaries}

The following result, which follows directly from \cite[Eq.\ (8)]{bareiss1968sylvesters}, will be instrumental in the following section.
Due to space limitation, we only present a sketch of its proof as a corollary of \cite[Eq.\ (8)]{bareiss1968sylvesters}.

\begin{proposition}[{from \cite{bareiss1968sylvesters}}]\label{pro-compute-determinant}
There is an algorithm that, given any $A\in\ZZb^\mxn$, computes the tuple $(r,\calR,\calC,\det(A_{[\calR,\calC]}))$ where (i) $r$ is the rank of $A$, (ii) $\calR=\{i_1,\ldots,i_r\}\subseteq\NNb$ with $1\leq i_1<i_2<\ldots<i_r\leq m$, (iii) $\calC=\{j_1,\ldots,j_r\}\subseteq\NNb$ with $1\leq j_1<j_2<\ldots<j_r\leq n$, and (iv) $\det\,A_{[\calR,\calC]}\neq0$.
Moreover, the bit complexity of the algorithm is polynomial in $\bitsize(A)$.
\end{proposition}

\begin{proof}
\extended{See Appendix \ref{ssec-pro-compute-determinant-proof} for a sketch of proof.}%
\short{See the extended version of this paper \cite{berger2021complexity}.}
\end{proof}

In particular, when $A\in\ZZb^\nxn$, the determinant of $A$ can be obtained from the output $(r,\calR,\calC,D)$ of the algorithm: if $r<n$, then $\det(A)=0$, otherwise $\det(A)=D$.

By combining the algorithm of Proposition \ref{pro-compute-determinant} with the well-known \emph{rule of Cramer} (see, e.g., \cite[\S0.8.3]{horn2013matrix}), one can obtain a polynomial-time algorithm for the resolution of systems of linear equations with integer coefficients.%
\footnote{Let us mention that more efficient algorithms for this problem have been proposed in the literature, such as the well-known \emph{Gaussian Elimination}.
However, the latter necessitates more advanced analysis for a careful proof of its polynomial-time nature (see \cite[\S2]{bareiss1968sylvesters}).}

\begin{proposition}\label{pro-solve-linear-system}
There is an algorithm that, given any $A\in\ZZb^\mxn$ and $b\in\ZZb^m$, computes integers $x_0\neq0$ and $x_1,\ldots,x_n$ such that $x=[x_1,\ldots,x_n]/x_0$ is a solution to $Ax=b$ if the system is feasible, or outputs that the system has no solution (in $\Re^n$).
Moreover, the bit complexity of the algorithm is polynomial in $\bitsize(A)+\bitsize(b)$.
\end{proposition}

\begin{proof}
\extended{See Appendix \ref{ssec-pro-solve-linear-system-proof}.}%
\short{See the extended version of this paper \cite{berger2021complexity}.}
\end{proof}

%%%%%%%%%%%%%%%%%%%%%%%%%%%%%%%%%%%%%%%%%%%%%%%%%%%%%%%%%%%%%%%%%%%%%%%%%%%%%%%%
\section{Proof of Theorem \ref{thm-boundedness-cont}}\label{sec-proof-thm-boundedness-cont}

%%%%%%%%%%%%%%%%%%%%%%%%%%%%%%%%%%%%%%%%%%%%%%%%%%%%%%%%%%%%%%%%%%%%%%%%%%%%%%%%
\subsection{The minimal polynomial}\label{ssec-minimal-polynomial}

The first step of our algorithm to answer Problem \ref{prob-boundedness-cont} is to compute the minimal polynomial of $A$.
We remind that the \emph{minimal polynomial} of a matrix $A\in\Re^\nxn$ is defined as the monic real polynomial $p:x\mapsto x^d + c_1x^{d-1} + \ldots + c_d$ with smallest degree $d$ such that $p(A)=0$.

The relevance of the minimal polynomial for Problem \ref{prob-boundedness-cont} is explained in Theorem \ref{thm-roots-minimal-polynomial-boundedness-cont} below.
First, we introduce the following terminology that will simplify the statement of the theorem.

\begin{definition}\label{def-boundedness-property}
A (complex or real) polynomial will be said to have the \emph{boundedness property} if each of its roots satisfies one of the following two conditions: (i) has negative real part, or (ii) is on the imaginary axis and is simple.
\end{definition}

\begin{theorem}[Folk]\label{thm-roots-minimal-polynomial-boundedness-cont}
For any $A\in\Re^\nxn$, it holds that $\sup_{t\in\Re_{\geq0}}\,\lVert e^{At}\rVert<\infty$ if and only if the minimal polynomial of $A$ has the boundedness property.
\end{theorem}

\begin{proof}
\extended{See Appendix \ref{ssec-thm-roots-minimal-polynomial-boundedness-cont-proof}.}%
\short{See the extended version of this paper \cite{berger2021complexity}.}
\end{proof}

The above theorem allows us to provide an algorithm to answer Problem \ref{prob-boundedness-cont}.
The algorithm is presented in Figure \ref{algo-prob-boundedness-cont}; it consists in two main steps that are described in the following subsections.

\begin{figure}
\hrule
\vskip3pt
\noindent{\bfseries Input:} $A\in\ZZb^\nxn$.

\noindent{\bfseries Output:} ``YES'' if $A$ is a positive instance of Problem \ref{prob-boundedness-cont} and ``NO'' otherwise.

\noindent\emph{Algorithm:}

\noindent$\triangleright$\kern3pt{\bfseries Step\kern2pt 1:}\kern3pt Using Theorem \ref{thm-compute-minimal-polynomial}, compute integers $e_0\neq0$ and $e_1,\ldots,e_d$ such that $x\mapsto x^d+e_1x^{d-1}/e_0+\ldots+e_d/e_0$ is the minimal polynomial of $A$.\\
Let $p:x\mapsto e_0x^d+e_1x^{d-1}+\ldots+e_d$.

\noindent$\triangleright$\kern3pt{\bfseries Step\kern2pt 2:}\kern3pt Using Theorem \ref{thm-analyze-roots-polynomial}, {\bfseries return} ``YES'' if $p$ has the boun\-dedness property (see Definition \ref{def-boundedness-property}) and {\bfseries return} ``NO'' other\-wise.
\vskip3pt
\hrule
\caption{Algorithm for answering Problem \ref{prob-boundedness-cont} for integer matrices (rational matrices can be treated in the same way by considering only the ``numerator matrix'').}
\label{algo-prob-boundedness-cont}
\end{figure}

%%%%%%%%%%%%%%%%%%%%%%%%%%%%%%%%%%%%%%%%%%%%%%%%%%%%%%%%%%%%%%%%%%%%%%%%%%%%%%%%
\subsection{Step 1: Computation of the minimal polynomial}\label{ssec-computing-minimal-poly}

The definition of the minimal polynomial, combined with the algorithm of Proposition \ref{pro-solve-linear-system}, allows for polynomial-time computation of the minimal polynomial of integer matrices.%
\footnote{Again, let us mention that more efficient algorithms have been proposed in the literature (see, e.g., in \cite{dumas2005efficient}), but necessitate more work for their description and for the analysis of their complexity.
Hence, we present an elementary algorithm to keep the paper simple and self-contained.}

\begin{theorem}\label{thm-compute-minimal-polynomial}
There is an algorithm that, given any $A\in\ZZb^\nxn$, computes integers $e_0\neq0$ and $e_1,\ldots,e_d$ such that $x\mapsto x^d+e_1x^{d-1}/e_0+\ldots+e_d/e_0$ is the minimal polynomial of $A$.
Moreover, the bit complexity of the algorithm is polynomial in $\bitsize(A)$.
\end{theorem}

\begin{proof}
For each $d\in\{1,\ldots,n\}$, write the matrix equa\-tion $A^d + c_1A^{d-1} + \ldots + c_d I=0$, with unknowns $c_1,\allowbreak\ldots,\allowbreak c_d\in\Re$.
For any $\ell\in\ZZb_{\geq0}$, it holds that the bit size of the entries of $A^\ell$ is bounded by $\ell\bitsize(A)\allowbreak+\ell\bitsize(n)$ (since each entry is the sum of $n^\ell$ products of $\ell$ elements of $A$).

The matrix equation can be rewritten in the classical vector form: $Mx=N$, where $x=[c_1,\ldots,c_d]$, $N=-\vecc(A^d)$ and $M=[\vecc(A^{d-1}),\ldots,\vecc(A^0)]$ ($\vecc(\cdot)$ is the vectorization operator%
\footnote{I.e., if $A=[a_1,\ldots,a_n]\in\Re^\mxn$ with $a_i\in\Re^m$ for all $i\in\{1,\ldots,n\}$, then $\vecc(A)=[a_1^\top,\ldots,a_n^\top]^\top$.}%
).
From the above, it holds that the bitsize of the entries of $M$ is bounded by $n\bitsize(A)+n^2\leq 2(\bitsize(A))^2$.
Hence, $\bitsize(M)\leq2(\bitsize(A))^5$ (since the number of elements of $M$ is equal to $n^2d$).
Similarly, we find that $\bitsize(N)\leq2(\bitsize(A))^4$.

Hence, using the algorithm of Proposition \ref{pro-solve-linear-system}, we can find integers $e_0\neq0$ and $e_1,\ldots,e_d$ such that $A^d+e_1A^{d-1}/e_0+\ldots+e_dI/e_0=0$ or conclude that no such numbers (integer or not) exist.
The smallest $d$ for which such integers exist provides the minimal polynomial of $A$.
Moreover, from the developments above, the bit size of these integers and the time to compute them is polynomial in $\bitsize(A)$.
\end{proof}

%%%%%%%%%%%%%%%%%%%%%%%%%%%%%%%%%%%%%%%%%%%%%%%%%%%%%%%%%%%%%%%%%%%%%%%%%%%%%%%%
\subsection{Step 2: Analysis of the roots of a polynomial}

The goal of this subsection is to prove Theorem \ref{thm-analyze-roots-polynomial} below, which states that deciding whether a polynomial with integer coefficients has the boundedness property can be done in polynomial time w.r.t.\ the bit size of its coefficients.

The proof relies on the \emph{Routh--Hurwitz stability criterion}, which is an algorithmic test to decide whether all the roots of a given polynomial have negative real part and was shown to be implementable by a polynomial-time algorithm (see, e.g., \cite{pena2004characterizations}).
Extensions of the Routh--Hurwitz criterion allow to compute for a given polynomial the number of roots on the imaginary axis and their multiplicity (see, e.g., \cite[\S15]{gantmacher2000thetheory}, \cite{chen1993new,choghadi2013routh}).
However, to the best of the authors' knowledge, no proof of the polynomial bit complexity of such algorithms is available in the literature.
Hence, in Theorem \ref{thm-analyze-roots-polynomial}, we present a minimalist version of the extended Routh--Hurwitz algorithm that is sufficient for our needs (verifying the boundedness property), and thriving on Proposition \ref{pro-compute-determinant}, we show that this minimalist version can be implemented by a polynomial-time algorithm.

\begin{theorem}\label{thm-analyze-roots-polynomial}
There is an algorithm that, given any polynomial $p:x\mapsto a_0x^d+\ldots+a_d$ with integer coefficients $a_0,\allowbreak\ldots,a_d$, outputs ``YES'' if $p$ has the boundedness property, and outputs ``NO'' otherwise.
Moreover, the bit complexity of the algorithm is polynomial in $\sum_{\ell=0}^d\bitsize(a_\ell)$.
\end{theorem}

The rest of this section is devoted to proving Theorem \ref{thm-analyze-roots-polynomial}.
To do that, we first introduce several results and concepts that are classical in the study of the Routh--Hurwitz criterion.

Let $p_0,p_1,\ldots,p_{m+1}$ be a sequence of real polynomials such that
\begin{equation}\label{eq-seq-poly}
\left\{\begin{array}{l}
p_{m+1}\equiv0,\quad p_k\not\equiv0, \quad \forall\,k\in\{1,\ldots,m\}, \\
p_{k+1} = -\remain(p_{k-1},p_k), \quad \forall\,k\in\{1,\ldots,m\},
\end{array}\right.
\end{equation}
where $\remain(p_{k-1},p_k)$ is the \emph{remainder} of the Euclidean div\-ision of $p_{k-1}$ by $p_k$, meaning that $\deg\,p_{k+1}<\deg\,p_k$ and there is a real polynomial $q_k$ such that $p_{k+1} = q_kp_k-p_{k-1}$.
Hence, the polynomials $p_0,\ldots,p_{m+1}$ are those that would be obtained by applying the \emph{Euclidean algorithm} (see, e.g., \cite[\S1.5]{cox2015ideals}) on $p_0$ and $p_1$, which is known to produce the GCD of $p_0$ and $p_1$.

\begin{lemma}[Folk]\label{lem-gcd-euclidean-algorithm}
Let $p_0,\ldots,\allowbreak p_{m+1}$ be as \eqref{eq-seq-poly}.
Then, for all $k\in\{0,\ldots,m\}$, $p_m$ is a greatest common divisor (GCD) of $p_k$ and $p_{k+1}$.
\end{lemma}

\begin{proof}
\extended{See Appendix \ref{ssec-lem-gcd-euclidean-algorithm-proof}.}%
\short{See the extended version of this paper \cite{berger2021complexity}.}
\end{proof}

From the above, it holds that $p_m$ divides $p_0$ and $p_1$.
The following result is known as the \emph{Routh--Hurwitz theorem}.

% From the above, it holds that $p_m$ divides each polynomial $p_0,\ldots,\allowbreak p_{m+1}$.
% For each $k\in\{0,\ldots,m\}$, define
% \begin{equation}\label{eq-red-poly}
% \pred_k=p_k/p_m.
% \end{equation}
% The following result is known as the \emph{Routh--Hurwitz theorem}.

\begin{lemma}[{see, e.g., \cite[Theorem 15.2]{gantmacher2000thetheory}}]\label{lem-routh-hurwitz-regular}
Let $p_0,\ldots,\allowbreak p_{m+1}$ be as in \eqref{eq-seq-poly}, and let $p:x\mapsto\pred_0(-ix)+i\pred_1(-ix)$ where $\pred_0=p_0/p_m$ and $\pred_1=p_1/p_m$.
Then, it holds that
\[
\tau_s-\tau_u = V(+\infty)-V(-\infty),
\]
where $\tau_s$ is the number of roots of $p$ with negative real part and $\tau_u$ is the number of roots of $p$ with positive real part, and $V(y)$ ($y\in\Re\cup\{\pm\infty\}$) is the number of variations of sign%
\footnote{The \emph{number of variations of sign} in a finite sequence of nonzero real numbers (or $\pm\infty$) is the number of pairs of consecutive elements in the sequence that have opposite sign.
For instance, the number of variations of sign in $1,\allowbreak\infty,\allowbreak-1,\allowbreak3,\allowbreak-\infty,\allowbreak-2$ is equal to $3$.}
in the sequence $p_0(y),\ldots,p_m(y)$.
\end{lemma}

\begin{proof}
\extended{See Appendix \ref{ssec-lem-routh-hurwitz-regular-proof}.}%
\short{See the extended version of this paper \cite{berger2021complexity}.}
\end{proof}

A similar approach can be used to compute the number of distinct real roots of a real polynomial.

\begin{lemma}[{see, e.g., \cite[p.\ 174]{gantmacher2000thetheory}}]\label{lem-sturm-real-roots}
Let $p_0,\ldots,\allowbreak p_{m+1}$ be as in \eqref{eq-seq-poly} with $p_1^{}=-p_0'$.
Then, $V(+\infty)-V(-\infty)$ is equal to the number of \emph{distinct real roots} of $p_0$, where $V(y)$ is the number of variations of sign in $p_0(y),\allowbreak\ldots,p_m(y)$.
\end{lemma}

\begin{proof}
\extended{See Appendix \ref{ssec-lem-sturm-real-roots-proof}.}%
\short{See the extended version of this paper \cite{berger2021complexity}.}
\end{proof}

By combining Lemmas \ref{lem-routh-hurwitz-regular} and \ref{lem-sturm-real-roots}, we obtain the following algorithmic procedure to decide whether a given polynomial has the boundedness property.

To introduce this procedure, let $p:x\mapsto c_0x^d+c_1x^{d-1}+\ldots+c_d$ be a real polynomial ($c_0\neq0$).
If $\deg\,p$ is even, then decompose $p$ into two real polynomials $p_0$ and $p_1$ such that $p(ix)=p_0(x)+ip_1(x)$ for all $x\in\Co$.
Namely,
\begin{equation}\label{eq-decomp-mp-even}
\begin{array}{l}
p_0:x\mapsto i^d c_0x^d + i^{d-2}c_2x^{d-2}+\ldots+c_d, \\
p_1:x\mapsto i^{d-2}c_1x^{d-1}+i^{d-4}c_3x^{d-3}+\ldots+c_{d-1}x.
\end{array}
\end{equation}
On the other hand, if $\deg\,p$ is odd, then decompose $p$ into two real polynomials $p_0$ and $p_1$ such that $ip(ix)=p_0(x)+ip_1(x)$ for all $x\in\Co$.
Namely,
\begin{equation}\label{eq-decomp-mp-odd}
\begin{array}{l}
p_0:x\mapsto i^{d+1} x^d + i^{d-1}c_2x^{d-2}+\ldots+i^2c_{d-1}x, \\
p_1:x\mapsto i^{d-1}c_1x^{d-1}+i^{d-3}c_3x^{d-3}+\ldots+c_d.
\end{array}
\end{equation}
In both cases, it holds that $\deg\,p_0>\deg\,p_1$.

Let $p_0,p_1,\ldots,p_{m+1}$ satisfy \eqref{eq-seq-poly} with $p_0$ and $p_1$ given by \eqref{eq-decomp-mp-even} or \eqref{eq-decomp-mp-odd}.
If $p_m$ is not a constant polynomial, then let%
\footnote{The superscript ``ext'' stands for ``extended'' because we extend the sequence $p_0,p_1,\ldots,p_{m+1}$.}
$\pext_0,\allowbreak\pext_1,\allowbreak\ldots,\pext_{\mext+1}$ satisfy \eqref{eq-seq-poly} with $\pext_0=p_m$ and $\pext_1=-p_m'$.
The following result links the boundedness property with the variations of sign in the sequences $p_0,p_1,\ldots,p_m$ and $\pext_0,\allowbreak\pext_1,\ldots,\pext_{\mext}$.

\begin{lemma}\label{lem-routh-hurwitz-boundedness}
Let $p_0,\ldots,p_{m+1}$ and $\pext_0,\allowbreak\ldots,\allowbreak\pext_{\mext+1}$ be as above.
Then, $p$ has the boundedness property if and only if
\begin{equation}\label{eq-sum-variations}
V(+\infty)-V(-\infty)+\Vext(+\infty)-\Vext(-\infty)=\deg\,p_0,
\end{equation}
where $V(y)$ and $\Vext(y)$ are the number of variations of sign in the sequences $p_0(y),\allowbreak\ldots,\allowbreak p_m(y)$ and $\pext_0(y),\allowbreak\ldots,\allowbreak\pext_{\mext}(y)$.
\end{lemma}

\begin{proof}
\extended{See Appendix \ref{ssec-lem-routh-hurwitz-boundedness-proof}.}%
\short{See the extended version of this paper \cite{berger2021complexity}.}
\end{proof}

From the above lemma, we obtain the following necessary condition for the satisfiability of the boundedness property.

\begin{corollary}\label{cor-routh-hurwitz-boundedness-degree-dec}
Let $p_0,\ldots,p_{m+1}$ and $\pext_0,\allowbreak\ldots,\allowbreak\pext_{\mext+1}$ be as above.
A \emph{necessary} condition for $p$ to have the boundedness property is that the degree difference between two consecutive polynomials is equal to one: i.e., $\deg\,p_k=\deg\,p_{k-1}-1$ for every $k\in\{1,\ldots,m\}$, and $\deg\,\pext_k=\deg\,\pext_{k-1}-1$ for every $k\in\{1,\ldots,\mext\}$.
\end{corollary}

\begin{proof}
\extended{See Appendix \ref{ssec-cor-routh-hurwitz-boundedness-degree-dec-proof}.}%
\short{See the extended version of this paper \cite{berger2021complexity}.}
\end{proof}

We are now able to prove Theorem \ref{thm-analyze-roots-polynomial} (see below).
For this, we use Lemma \ref{lem-routh-hurwitz-boundedness}, which requires to compute $V(+\infty)-V(-\infty)$ and $\Vext(+\infty)-\Vext(-\infty)$.%
\footnote{The naive way to do this would be to compute the polynomials $p_0,\allowbreak\ldots,\allowbreak p_m$ and $\pext_0,\allowbreak\ldots,\allowbreak\pext_{\mext}$, and look at the variations of sign in the associated sequences.
However, a proof of Theorem \ref{thm-analyze-roots-polynomial} based on this would require to show that the computation of these polynomials can be done in polynomial time, which is long and tedious (see, e.g., \cite[\S6]{von2013modern}).
Therefore, we use another approach, based on the ``Hurwitz determinants''.}
We will see that this can be achieved by computing the determinants of matrices built from the coefficients of $p_0$ and $p_1$.
This is the idea of the ``Hurwitz determinants'' obtained from the ``Hurwitz matrix'' (see, e.g., \cite[\S15.6]{gantmacher2000thetheory}).
By combining it with Proposition \ref{pro-compute-determinant}, we deduce that this can be done in polynomial time.

\begingroup
\renewcommand{\myproofname}{Proof of Theorem \ref{thm-analyze-roots-polynomial}}
\begin{proof}
First, we explain how to compute $V(+\infty)\allowbreak-V(-\infty)$ and $\pext_0$; then we apply the exact same idea to compute $\Vext(+\infty)-\Vext(-\infty)$.

Let $p_0$ and $p_1$ be as in \eqref{eq-decomp-mp-even} or \eqref{eq-decomp-mp-odd}.
For definiteness, suppose that we are in the case of \eqref{eq-decomp-mp-even} (i.e., the degree of $p$ is even); the case of \eqref{eq-decomp-mp-odd} is exactly the same.
Let $d=\deg\,p=2f$.
If $\deg\,p_1<d-1$, then $p$ does not satisfy the necessary condition of Corollary \ref{cor-routh-hurwitz-boundedness-degree-dec} so that there is no need for further computations.
Thus, we assume that $\deg\,p_1=d-1$.
Denote the coefficients of $p_0$ and $p_1$ by%
\footnote{In the rest of this subsection, for the sake of readability, superscripts are used both as exponents and as indexes, but the distinction should be clear from the context; e.g., $a^0_0$, $\calM^1$ (\emph{index}) vs.\ $x^d$ (\emph{exponent}).}
\begin{equation}\label{eq-coeffs-p0-p1-hurwitz}
\begin{array}{l}
p_0:x\mapsto a^0_0x^d+a^0_1x^{d-2}+\ldots+a^0_f, \\
p_1:x\mapsto a^1_0x^{d-1}+a^1_1x^{d-3}+\ldots+a^1_{f-1}x.
\end{array}
\end{equation}
Consider the following $(d+1)\times(d+1)$ matrix:
% \begin{equation}\label{eq-hurwitz-matrix}
\[
\scalebox{0.85}{$
\renewcommand{\arraystretch}{1.35}
\calM=\left[\begin{array}{@{\;}c*{9}{@{\hspace{5pt}}c}}
a^0_0 & a^0_1 & a^0_2 & {\;\cdots\;} & a^0_f \\
& a^1_0 & a^1_1 & \cdots & a^1_{f-1} \\
& a^0_0 & a^0_1 & \cdots & a^0_{f-1} & a^0_f \\
& & a^1_0 & \cdots & a^1_{f-1} & a^1_{f-1} \\
& & a^0_0 & \cdots & a^0_{f-2} & a^0_{f-1} & a^0_f \\
& & & \ddots & & & & \ddots \\
& & & & a^1_0 & a^1_1 & a^1_2 & {\;\cdots\;} & a^1_{f-1} \\
& & & & a^0_0 & a^0_1 & a^0_2 & \cdots & a^0_{f-1} & a^0_f \\
\end{array}\!\right].$}
\]
% \end{equation}
Let $q_1:x\mapsto b_1x$ be such that $\deg(q_1p_1-p_0)<\deg\,p_1$.
By \eqref{eq-coeffs-p0-p1-hurwitz}, it is equivalent to asking that $b_1^{}a^1_0-a^0_0=0$.
Hence, if, for each $k\in\{3,5,7,\ldots,d+1\}$, we transform $\calM_{[k,:]}$ (the $k$th row of $\calM$) into $b_1\calM_{[k-1,:]}-\calM_{[k,:]}$, then we get the following matrix:
\[
\scalebox{0.85}{$
\renewcommand{\arraystretch}{1.35}
\calM^1=\left[\begin{array}{@{\;}c*{11}{@{\hspace{5pt}}c}}
a^0_0 & a^0_1 & a^0_2 & a^0_3 & {\;\cdots\;} & a^0_f \\
& a^1_0 & a^1_1 & a^1_2 & \cdots & a^1_{f-1} \\
& & a^2_0 & a^2_1 & \cdots & a^2_{f-2} & a^0_{f-1} \\
& & a^1_0 & a^1_1 & \cdots & a^1_{f-1} & a^1_{f-1} \\
& & & a^2_0 & \cdots & a^2_{f-2} & a^0_{f-1} \\
& & & & \ddots & & & & \ddots \\
& & & & & a^1_0 & a^1_1 & a^1_2 & {\;\cdots\;} & a^1_{f-1} \\
& & & & & & a^2_0 & a^2_1 & \cdots & a^2_{f-2} & a^0_{f-1} \\
\end{array}\!\right].$}
\]
Any $k$th row of $\calM^1$, with $k\in\{3,5,7,\ldots,d+1\}$, gives the coefficients of the polynomial $p_2$, defined by $p_2=q_1p_1-p_0$.
Namely, $p_2:x\mapsto a^2_0x^{d-2}+a^2_1x^{d-4}+\ldots+a^2_{f-1}$.
The interest of this approach is that to compute the sign of the coefficients $a^2_0,\ldots,a^2_{f-1}$, we do not need to compute $\calM^1$, it suffices to compute the determinant of a submatrix of $\calM$.
More precisely, for all $\ell\in\{0,\ldots,f-1\}$, it holds that
\[
\det\,\calM^1_{[1:3,1:2\cup\{\ell+3\}]}=a^0_0a^1_0a^2_\ell=-\det\,\calM_{[1:3,1:2\cup\{\ell+3\}]},
\]
since $\calM_{[1:3,:]}^1$ was obtained from $\calM_{[1:3,:]}^{}$ by using the row transformation $\calM^1_{[3,:]}\coloneqq b_1\calM_{[2,:]}-\calM_{[3,:]}$.

In the same way as above, for each $k\in\{4,6,8,\ldots,d\}$, we can eliminate the first element of $\calM^1_{[k,:]}$ by transforming $\calM^1_{[k,:]}$ into $b_2\calM^1_{[k-1,:]}-\calM^1_{[k,:]}$, where $q_2:x\mapsto b_2x$ is such that $\deg(q_2p_2-p_1)<\deg\,p_1$.
This will give a matrix $\calM^2$ containing \emph{among others} the coefficients of the polynomial $p_3:x\mapsto a^3_0x^{d-3}+a^3_1x^{d-5}+\ldots+a^3_{f-2}$ defined by $p_3=q_2p_2-p_1$.
From this, we get that, for all $\ell\in\{0,\ldots,f-2\}$,
\[
a^0_0a^1_0a^2_0a^3_\ell=\det\,\calM_{[1:4,1:3\cup\{\ell+4\}]}.
\]
By using the same reasoning inductively, we get the general relation: for all $k\in\{0,\ldots,d\}$ and $\ell\in\{0,\ldots,f-\lceil k/2\rceil\}$,
\begin{equation}\label{eq-coeffs-Hurwitz-determinants}
a^0_0\cdots a^{k-1}_0a^k_\ell = \sigma_k \det\,\calM_{[1:k+1,1:k\cup\{\ell+k+1\}]},
\end{equation}
where $\sigma_k=-1$ if $k\in4\ZZb+2$ and $\sigma_k=1$ otherwise, and $p_k:x\mapsto a^k_0x^{d-k}+a^k_1x^{d-k-2}+\ldots+a^k_{f-\lceil k/2\rceil}x^{k\,\mathrm{mod}\,2}$ is the $k$th polynomial in the sequence obtained by \eqref{eq-seq-poly} with $p_0$ and $p_1$ as in \eqref{eq-coeffs-p0-p1-hurwitz}.

If $\det\,M_{[1:k+1,1:k+1]}=0$ for some $k\in\{1,\ldots,d\}$, this means that $a^k_0=0$.
In this case, two situations can occur:
\begin{itemize}[\nocalcleftmargintrue\leftmargin=25pt]
\item[(S1)] $\det\,\calM_{[1:k+1,1:k\cup\{\ell+k+1\}]}=0$ for all $\ell\in\{1,\ldots,f-\lceil k/2\rceil\}$.
This means that $p_k\equiv0$.
\item[(S2)] $\det\,\calM_{[1:k+1,1:k\cup\{\ell+k+1\}]}\neq0$ for some $\ell\in\{1,\ldots,\allowbreak f-\lceil k/2\rceil\}$.
This means that $p_k\not\equiv0$ but $\deg\,p_k<\deg\,p_{k-1}-1$.
\end{itemize}

The above leads to the following algorithm for the computation of $V(+\infty)-V(-\infty)+\Vext(+\infty)-\Vext(-\infty)$.

\underline{\smash{Algorithm}:} Using Proposition \ref{pro-compute-determinant}, we compute the determinant of $\calM_{[1:k+1,1:k+1]}$ for $k=0,1,\ldots,d$.
If the determinant is nonzero for all $k\in\{0,\ldots,d\}$, then we deduce from \eqref{eq-coeffs-Hurwitz-determinants} the signs of the leading coefficients $a^0_0,\ldots,a^d_0$.
We verify whether these signs are strictly \emph{alternating} since this is the only way to have $V(+\infty)-V(-\infty)=d$.
If this is the case, we stop the algorithm and output ``YES'' since $p$ satisfies \eqref{eq-sum-variations} in Lemma \ref{lem-routh-hurwitz-boundedness}.
Otherwise, we stop the algorithm and output ``NO'' since $p$ does not satisfy \eqref{eq-sum-variations} in Lemma \ref{lem-routh-hurwitz-boundedness}.

On the other hand, if, for some $k\in\{1,\ldots,d\}$, the determinant of $\calM_{[1:k+1,1:k+1]}$ is zero, then we check whether we are in situation (S1) or (S2) above.
If we are in (S2), then it means that $p$ does not satisfy the necessary condition of Corollary \ref{cor-routh-hurwitz-boundedness-degree-dec}, and thus, we stop the algorithm and output ``NO''.
Otherwise (we are in (S1)), we let $m$ be the smallest $k$ such that $\det\,\calM_{[1:k+2,1:k+2\}]}=0$, and from \eqref{eq-coeffs-Hurwitz-determinants} we compute the sign of $a^0_0,\ldots,a^m_0$ (the leading coefficients of $p_0,\ldots,p_m$).
If the signs are not strictly \emph{alternating}, then we stop the algorithm and output ``NO'' since $p$ does not satisfy \eqref{eq-sum-variations} in Lemma \ref{lem-routh-hurwitz-boundedness}.
Otherwise, from \eqref{eq-coeffs-Hurwitz-determinants} with $k=m$ and $\ell=0,\ldots,f-\lceil m/2\rceil$, we define $\pext_0=\lvert a^0_0\cdots a^{m-1}_0\rvert p_m$.

At this stage, if the algorithm did not stop and outputted ``Yes'' or ``NO'', then it produced the polynomial $\pext_0=\alpha p_m$ with $\alpha=\lvert a^0_0\cdots a^{m-1}_0\rvert>0$.
From $\pext_0$, we compute the value of $\Vext(+\infty)-\Vext(-\infty)$ as follows.
First, we define $\pext_1=-(\pext_0)'$.
Then, in the same way as above, we compute the signs of $a^{\text{ext},0}_0,\ldots,a^{\text{ext},\mext}_0$ (the leading coefficients of $\pext_0,\allowbreak\ldots,\pext_{\mext}$).
If $\mext=d-m$ and the signs are strictly \emph{alter\-nating}, then we stop the algorithm and output ``YES'' since $p$ satisfies \eqref{eq-sum-variations} in Lemma \ref{lem-routh-hurwitz-boundedness}.
Otherwise, we stop the algorithm and output ``NO'' since $p$ does not satisfy \eqref{eq-sum-variations} in Lemma \ref{lem-routh-hurwitz-boundedness}.\hfill~$\triangleleft$

The above algorithm requires at most $d^2$ computations of the determinant of a submatrix of $\calM$.
By Proposition \ref{pro-compute-determinant}, these determinants can be computed in polynomial time w.r.t.\ the bit size of the entries of $\calM$.
Since these entries consist in the coefficients of the input polynomial $p$, this concludes the proof of the theorem.
\end{proof}
\endgroup

%%%%%%%%%%%%%%%%%%%%%%%%%%%%%%%%%%%%%%%%%%%%%%%%%%%%%%%%%%%%%%%%%%%%%%%%%%%%%%%%
\section{Proof of Theorem \ref{thm-boundedness-disc}}\label{sec-proof-thm-boundedness-disc}

The polynomial-time algorithm to answer Problem \ref{prob-boundedness-cont} presented in the previous section (see Figure \ref{algo-prob-boundedness-cont}) can be easily adapted, by adding an intermediate step between Step 1 and Step 2, to obtain a polynomial-time algorithm for Problem \ref{prob-boundedness-disc}.
The intermediate step consists in a transformation of the min\-imal polynomial of the matrix, called a \emph{M\"obius transformation}, which maps the interior of the unit circle in the complex plane to the interior of the left-hand side plane.
The relevance of this transformation is explained in Theorem \ref{thm-roots-minimal-polynomial-boundedness-disc} below.

\begin{definition}\label{def-boundedness-property-disc}
A (complex or real) polynomial will be said to have the \emph{discrete-time boundedness property} if each of its roots satisfies one of the following two conditions: (i) is in the interior of the unit circle, or (ii) is on the unit circle and is simple.
\end{definition}

\begin{theorem}[Folk]\label{thm-roots-minimal-polynomial-boundedness-disc}
For any $A\in\Re^\nxn$, it holds that $\sup_{t\in\ZZb_{\geq0}}\,\lVert A^t\rVert<\infty$ if and only if the minimal polynomial of $A$ has the discrete-time boundedness property.
\end{theorem}

\begin{proof}
\extended{See Appendix \ref{ssec-thm-roots-minimal-polynomial-boundedness-disc-proof}.}%
\short{See the extended version of this paper \cite{berger2021complexity}.}
\end{proof}

\begin{theorem}\label{thm-analyze-roots-polynomial-disc}
There is an algorithm that, given any polynomial $\pdisc:x\mapsto a_0x^d+\ldots+a_d$ with integer coefficients $a_0,\allowbreak\ldots,\allowbreak a_d$, outputs a polynomial $p$ such that $\pdisc$ has the discrete-time boundedness property if and only if $p$ has the boundedness property.
Moreover, the bit complexity of the algorithm is polynomial in $\sum_{\ell=0}^d\bitsize(a_\ell)$.
\end{theorem}

\begin{proof}
\extended{See Appendix \ref{ssec-thm-analyze-roots-polynomial-disc-proof}.}%
\short{See the extended version of this paper \cite{berger2021complexity}.}
\end{proof}

Putting things together, we get the polynomial-time algorithm presented in Figure \ref{algo-prob-boundedness-disc} to answer Problem \ref{prob-boundedness-disc}.

\begin{figure}
\hrule
\vskip3pt
\noindent{\bfseries Input:} $A=B/q$, with $A\in\ZZb^\nxn$ and $q\in\ZZb_{>0}$.

\noindent{\bfseries Output:} ``YES'' if $A$ is a positive instance of Problem \ref{prob-boundedness-disc} and ``NO'' otherwise.

\noindent\emph{Algorithm:}

\noindent$\triangleright$\kern3pt{\bfseries Step\kern2pt 1:}\kern3pt Using Theorem \ref{thm-compute-minimal-polynomial}, compute integers $e_0\neq0$ and $e_1,\ldots,e_d$ such that $x\mapsto x^d+e_1x^{d-1}/e_0+\ldots+e_d/e_0$ is the minimal polynomial of $B$.\\
Let $\pdisc:x\mapsto e_0q^dx^d+e_1q^{d-1}x^{d-1}+\ldots+e_1q+e_d$.

\noindent$\triangleright$\kern3pt{\bfseries Inter-step:}\kern3pt Using Theorem \ref{thm-analyze-roots-polynomial-disc}, compute a polynomial $p$ that has the boundedness property if and only if $\pdisc$ has the discrete-time boundedness property.

\noindent$\triangleright$\kern3pt{\bfseries Step\kern2pt 2:}\kern3pt Using Theorem \ref{thm-analyze-roots-polynomial}, {\bfseries return} ``YES'' if $p$ has the boun\-dedness property and {\bfseries return} ``NO'' other\-wise.
\vskip3pt
\hrule
\caption{Algorithm for answering Problem \ref{prob-boundedness-disc}.}
\label{algo-prob-boundedness-disc}
\end{figure}

%%%%%%%%%%%%%%%%%%%%%%%%%%%%%%%%%%%%%%%%%%%%%%%%%%%%%%%%%%%%%%%%%%%%%%%%%%%%%%%%
\section{Conclusions}

Summarizing, in this paper, we showed that the problem of deciding whether a linear time invariant dynamical system, with rational transition matrix, has bounded trajectories can be answered in polynomial time with respect to the bit size of the entries of the transition matrix.
To do this, we leveraged several tools from system and control theory and from computer algebra, and we provided a careful analysis of the computational complexity of these tools when integrated into a complete algorithm for our decision problem.

For further work, it would interesting to derive tight upper bounds on the complexity of the described algorithm (and of some improved versions not presented here to keep the paper simple and self-contained), and also to compare it with the complexity that could be obtained with other types of algorithms, like randomized algorithms, which are known to provide practically efficient algorithms, for instance, for the computation of the determinant of integer matrices, or for the computation of the GCD of polynomials with integer coefficients.

%%%%%%%%%%%%%%%%%%%%%%%%%%%%%%%%%%%%%%%%%%%%%%%%%%%%%%%%%%%%%%%%%%%%%%%%%%%%%%%%%%%%%%%%%%%%%%%%%%%%%%%%%%%%%%%%%%%%%%%%%%%%%%%%%%%%%%%%%%%%%%%%%%%%%%%%%%%%%%%%
% 
% 
% 
% 
% 
% 
% 
% 
% 
% 
% 
% 
% 
% 
% 
% 
% 
% 
% 
%%%%%%%%%%%%%%%%%%%%%%%%%%%%%%%%%%%%%%%%%%%%%%%%%%%%%%%%%%%%%%%%%%%%%%%%%%%%%%%%%%%%%%%%%%%%%%%%%%%%%%%%%%%%%%%%%%%%%%%%%%%%%%%%%%%%%%%%%%%%%%%%%%%%%%%%%%%%%%%%
\bibliographystyle{plain}
\bibliography{myrefs}

\begin{thebibliography}{10}

\bibitem{bareiss1968sylvesters}
Erwin~H Bareiss.
\newblock {Sylvester's} identity and multistep integer-preserving {Gaussian}
  elimination.
\newblock {\em Mathematics of Computation}, 22:565--578, 1968.

\bibitem{bartels1972algorithm}
Richard~H. Bartels and George~W Stewart.
\newblock Algorithm 432: solution of the matrix equation {$AX+ XB= C$}.
\newblock {\em Communications of the ACM}, 15(9):820--826, 1972.

\bibitem{chen1993new}
Shyan~S Chen and Jason~SH Tsai.
\newblock A new tabular form for determining root distribution of a complex
  polynomial with respect to the imaginary axis.
\newblock {\em IEEE transactions on automatic control}, 38(10):1536--1541,
  1993.

\bibitem{choghadi2013routh}
Mohammad~Amin Choghadi and Heidar~A Talebi.
\newblock The routh-hurwitz stability criterion, revisited: the case of
  multiple poles on imaginary axis.
\newblock {\em IEEE Transactions on Automatic Control}, 58(7):1866--1869, 2013.

\bibitem{cox2015ideals}
David~A Cox, John Little, and Donal O'Shea.
\newblock {\em Ideals, varieties, and algorithms: an introduction to
  computational algebraic geometry and commutative algebra}.
\newblock Springer, Cham, 4\textsuperscript{th} edition, 2015.

\bibitem{dumas2006bounds}
Jean-Guillaume Dumas.
\newblock Bounds on the coefficients of the characteristic and minimal
  polynomials.
\newblock {\em arXiv preprint cs/0610136}, 2006.

\bibitem{dumas2005efficient}
Jean-Guillaume Dumas, Cl{\'e}ment Pernet, and Zhendong Wan.
\newblock Efficient computation of the characteristic polynomial.
\newblock In {\em Proceedings of the 2005 International Symposium on Symbolic
  and Algebraic Computation}, pages 140--147. ACM, 2005.

\bibitem{gantmacher2000thetheory}
Felix~R Gantmacher.
\newblock {\em The theory of matrices, {Vol. 2}}.
\newblock American Mathematical Society, Providence, RI, 2000.

\bibitem{horn2013matrix}
Roger~A Horn and Charles~R Johnson.
\newblock {\em Matrix analysis}.
\newblock Cambridge University Press, Cambridge, MA, 2\textsuperscript{nd}
  edition, 2013.

\bibitem{hurwitz1895ueber}
Adolf Hurwitz.
\newblock Ueber die {Bedingungen}, unter welchen eine {Gleichung} nur {Wurzeln}
  mit negativen reellen {Teilen} besitzt.
\newblock {\em Mathematische Annalen}, 46:273--284, 1895.

\bibitem{lienard1914surlesigne}
Alfred-Marie Li\'enard and Henri Chipart.
\newblock Sur le signe de la partie r\'eelle des racines d'une \'equation
  alg\'ebrique.
\newblock {\em Journal de Math\'ematiques Pures et Appliqu\'ees}, 10:291--346,
  1914.

\bibitem{lind1995anintroduction}
Douglas Lind and Brian Marcus.
\newblock {\em An introduction to symbolic dynamics and coding}.
\newblock Cambridge University Press, Cambridge, UK, 1995.

\bibitem{pena2004characterizations}
Juan~M Pe{\~n}a.
\newblock Characterizations and stable tests for the {Routh--Hurwitz}
  conditions and for total positivity.
\newblock {\em Linear Algebra and its Applications}, 393:319--332, 2004.

\bibitem{routh1877treatise}
Edward~John Routh.
\newblock {\em A treatise on the stability of a given state of motion}.
\newblock Macmillan, London, 1877.

\bibitem{von2013modern}
Joachim von~zur Gathen and J{\"u}rgen Gerhard.
\newblock {\em Modern computer algebra}.
\newblock Cambridge University Press, New York, NY, 3\textsuperscript{rd}
  edition, 2013.

\end{thebibliography}

\extended{%

\appendix

%%%%%%%%%%%%%%%%%%%%%%%%%%%%%%%%%%%%%%%%%%%%%%%%%%%%%%%%%%%%%%%%%%%%%%%%%%%%%%%%
\subsection{Sketch of proof of Proposition \ref{pro-compute-determinant}}\label{ssec-pro-compute-determinant-proof}

For the simplicity of notation and without loss of generality%
\footnote{For instance, it suffices to fill in any rectangular matrix with rows or columns of zeros to make it square without changing the assertions of the proposition.}%
, we assume that the matrix is square: $A=(a_{ij})_{i=1,j=1}^{n,n}\in\ZZb^\nxn$.
The algorithm (adapted from \cite{bareiss1968sylvesters}) works as follows.
First, we define $a_{00}^{(-1)} = 1$ and $a_{ij}^{(0)}=a_{ij}$ for each $i,j\in\{1,\allowbreak\ldots,n\}$.
Then, for $k=1,2,\ldots,n$, we define recursively
\begin{equation}\label{eq-def-rec}
a_{ij}^{(k)} = \big(a_{kk}^{(k-1)}a_{ij}^{(k-1)}-a_{kj}^{(k-1)}a_{ik}^{(k-1)}\big)\big/a_{k-1,k-1}^{(k-2)}
\end{equation}
for all $i,j\in\{k+1,\ldots,n\}$.
The formula \eqref{eq-def-rec} is well defined as long as $a_{k-1,k-1}^{(k-2)}\neq0$.
Thus, if for some $k\in\{1,\ldots,n\}$, $a_{k+1,k+1}^{(k)}=0$, then we look whether there are indices $i,j\in\{k+2,\ldots,n\}$ such that $a_{ij}^{(k)}\neq0$.
Two situations can occur: (i) such $i,j$ exist, or (ii) no such $i,j$ exist.

If we are in situation (i), then we define a permutation of the row indices such that the $i$th index becomes the $(k+1)$st index and a permutation of the column indices such that the $j$th index becomes the $(k+1)$st index.
Note that these permutations of the indices do not affect that values of $a_{ij}^{(\ell)}$ for $\ell\in\{-1,\ldots,k-1\}$ and $i,j\in\{\ell+1,\ldots,k\}$.
Hence, we can resume the recurrence \eqref{eq-def-rec} with the new indexing, which satisfies that $a_{k+1,k+1}^{(k)}\neq0$.

On the other hand, if we are in situation (ii), then it means that $k$ is equal to the rank of $A$ and thus we let $r=k$ and we stop the recurrence \eqref{eq-def-rec}.

At this stage of the algorithm, the recurrence equation \eqref{eq-def-rec} holds for all $k\in\{1,\ldots,r\}$ and all $i,j\in\{k+1,\ldots,n\}$ (with the reordering of the indices computed during the recurrence; see situation (i) above).
It is shown in \cite[\S1]{bareiss1968sylvesters} that the iterates of the recurrence \eqref{eq-def-rec} satisfy
\begin{equation}\label{eq-det-recurrence}
a_{ij}^{(k)}=\det\,A_{[1:k\cup\{i\},1:k\cup\{j\}]}.
\end{equation}
for all $k\in\{1,\ldots,r\}$ and all $i,j\in\{k+1,\ldots,n\}$.
Thus, we let $\calR$ be the first $r$ indices (with the reordering of the row indices computed during the recurrence), $\calC$ be the first $r$ indices (with the reordering of the column indices computed during the recurrence) and $D=a_{rr}^{(r-1)}$.

The above discussion shows that the rank of $A$, the subsets $\calR$ and $\calC$, and the determinant $D\coloneqq\det\,A_{[\calR,\calC]}\neq0$ can be computed with a number of arithmetic operations polynomial in $n$, using the recurrence \eqref{eq-def-rec}.
Moreover, by \eqref{eq-det-recurrence}, it holds that the intermediate integers $a_{ij}^{(k)}$ involved in the recurrence are equal to the determinant of submatrices of $A$, and thus their bit is polynomial in $\bitsize(A)$ since the bit size of the determinant of a $k\times k$ submatrix of $A$ is bounded by $k\bitsize(A)+k\bitsize(k)\leq2(\bitsize(A))^2$ (as it is the sum of $k!$ products of $k$ elements of $A$).
This concludes the proof of the proposition.\hspace*{\fill}~\QED

% It is shown in \cite[\S1]{bareiss1968sylvesters} that the recurrence defined by $a_{00}^{(-1)} = 1$, $a_{ij}^{(0)}=a_{ij}$ for all $i,j\in\{1,\allowbreak\ldots,n\}$, and
% , as long as $k\leq n$ and $a_{k-1,k-1}^{(k-2)}\neq0$, satisfies that
% \begin{equation}\label{eq-det-recurrence}
% a_{ij}^{(k)}=\det\,A_{[1:k\cup\{i\},1:k\cup\{j\}]}.
% \end{equation}

%%%%%%%%%%%%%%%%%%%%%%%%%%%%%%%%%%%%%%%%%%%%%%%%%%%%%%%%%%%%%%%%%%%%%%%%%%%%%%%%
\subsection{Proof of Proposition \ref{pro-solve-linear-system}}\label{ssec-pro-solve-linear-system-proof}

The algorithm works as follows.
Using the algorithm of Proposition \ref{pro-compute-determinant}, we compute $(r,\calR,\calC,D)$ where $r$ is the rank of $A$ and $D=\det\,A_{[\calR,\calC]}$.
For the simplicity of notation, we will assume that $\calR=\calC=\{1,\ldots,r\}$.
Then, it holds that any column of $A$ with index $j>r$ is a linear combination of the first $r$ columns of $A$.
Thus, the system $Ax=b$ has a solution $x\in\Re^n$ if and only if $A_{[:,1:r]}y=b$ has a solution $y\in\Re^r$.

Since $D\neq0$, the system $A_{[1:r,1:r]}y=b_{[1:r]}$ has a unique solution $y\in\Re^r$.
Moreover, this solution can be computed in polynomial time, using Cramer's rule (see, e.g., \cite[\S0.8.3]{horn2013matrix}) and the algorithm of Proposition \ref{pro-compute-determinant}: we let $y_0=D$ and
\[
y_j=\det(A_{[1:r,1:r]}\!\xleftarrow{j}\!b_{[1:r]}),\quad j\in\{1,\ldots,r\},
\]
where $A_{[1:r,1:r]}\xleftarrow{j}\allowbreak b_{[1:r]}$ is the matrix $A_{[1:r,1:r]}$ with its $j$th column replaced by $b_{[1:r]}$; then $y=[y_1,\ldots,y_r]/y_0$ is the unique solution of $A_{[1:r,1:r]}y=b_{[1:r]}$.
Hence, if $A_{[:,1:r]}y=b$, then the integers $x_0,x_1,\ldots,x_n$, defined by $x_i=y_i$ if $i\leq r$ and $x_i=0$ otherwise, satisfy the assertions of the corollary.
Otherwise (if $A_{[:,1:r]}y\neq b$), it means that the system $Ax=b$ has no solution in $\Re^n$.

In total, we have computed $r$ determinants and one matrix-vector multiplication.
Since each operation can be computed in time polynomial time w.r.t.\ $\bitsize(A)+\bitsize(b)$ (see Proposition \ref{pro-compute-determinant} for the determinant; the case of matrix-vector multiplication is trivial) and since $r\leq\min(m,n)\leq\bitsize(A)$, we have that the total complexity of the algorithm is polynomial in $\bitsize(A)+\bitsize(b)$, which concludes the proof.\hspace*{\fill}~\QED

%%%%%%%%%%%%%%%%%%%%%%%%%%%%%%%%%%%%%%%%%%%%%%%%%%%%%%%%%%%%%%%%%%%%%%%%%%%%%%%%
\subsection{Proof of Theorem \ref{thm-roots-minimal-polynomial-boundedness-cont}}\label{ssec-thm-roots-minimal-polynomial-boundedness-cont-proof}
The proof relies on the following well-known property of the minimal polynomial (see, e.g., \cite[Theorem 3.3.6]{horn2013matrix}).

\underline{\smash{Link between m.p. and eigenvalues:}}
Let $A\in\Re^\nxn$.
The minimal polynomial of $A$ is equal to $x\mapsto\prod_{s=1}^m(x-\lambda_s)^{n_s}$, where $\lambda_1,\ldots,\allowbreak\lambda_m$ are the distinct eigenvalues of $A$ and $n_s$ is the size of the largest Jordan block associated to $\lambda_s$ in the Jordan canonical form of $A$.\hfill~$\triangleleft$

Then, using the above property, the conclusion of Theorem \ref{thm-roots-minimal-polynomial-boundedness-cont} follows from the expression of the exponential of a Jordan block $J_n(\lambda)$: $\euler^{J_n(\lambda)t}=\euler^{\lambda t}\sum_{k=0}^{n-1}\frac{1}{k!}(J_n(0)t)^k$.\hspace*{\fill}~\QED

%%%%%%%%%%%%%%%%%%%%%%%%%%%%%%%%%%%%%%%%%%%%%%%%%%%%%%%%%%%%%%%%%%%%%%%%%%%%%%%%
\subsection{Proof of Lemma \ref{lem-gcd-euclidean-algorithm}}\label{ssec-lem-gcd-euclidean-algorithm-proof}

It is clear that $p_m$ is a GCD of $p_m$ and $p_{m+1}$.
For a proof by contradiction, assume that the statement of the lemma is false and let $k$ be the largest integer such that $p_m$ is not a GCD of $p_k$ and $p_{k+1}$.
Since $p_k=q_{k+1}p_{k+1}-p_{k+2}$ and $p_m$ is a GCD of $p_{k+1}$ and $p_{k+2}$, it holds that $p_m$ divides $p_k$.
Thus, $p_m$ is a common divisor of $p_k$ and $p_{k+1}$.
On the other hand, for the same reason, any common divisor of $p_k$ and $p_{k+1}$ will also divide $p_{k+2}$, and thus it will be a common divisor of $p_{k+1}$ and $p_{k+2}$ so that it will also divide their GCD $p_m$.
Hence, $p_m$ is a GCD of $p_k$ and $p_{k+1}$, a contradiction with the definition of $k$, concluding the proof.\hspace*{\fill}~\QED

%%%%%%%%%%%%%%%%%%%%%%%%%%%%%%%%%%%%%%%%%%%%%%%%%%%%%%%%%%%%%%%%%%%%%%%%%%%%%%%%
\subsection{Proof of Lemma \ref{lem-routh-hurwitz-regular}}\label{ssec-lem-routh-hurwitz-regular-proof}

Let $p_0,\ldots,\allowbreak p_{m+1}$ be as in the statement of the lemma and for each $k\in\{0,\ldots,m+1\}$, define $\pred_k=p_k/p_m$.
Let $\Vred(y)$ be the number of variations of sign in the sequence $\pred_0(y),\allowbreak\ldots,\allowbreak\pred_m(y)$.

\begin{lemma}\label{lem-variations-reduced}
Let $p_0,\ldots,\allowbreak p_{m+1}$ and $\pred_0,\allowbreak\ldots,\allowbreak\pred_{m+1}$ be as above.
It holds that $V(+\infty)-V(-\infty)=\Vred(+\infty)-\Vred(-\infty)$.
\end{lemma}

\begin{proof}
For any $y\in\Re\cup\{\pm\infty\}$ such that $p_m(y)\neq0$, multiplying each $\pred_k(y)$ ($k\in\{1,\ldots,m\}$) by $p_m(y)$ does not change the number of variations of sign in the sequence $\pred_0(y),\ldots,\pred_m(y)$: i.e., $V(y)=\Vred(y)$.
In particular, it holds that $V(+\infty)-V(-\infty)=\Vred(+\infty)-\Vred(-\infty)$.
\end{proof}

\begin{lemma}\label{lem-variation-argument-roots}
Let $p$ be a (real or complex) polynomial with no root on the imaginary axis.
Then,
\[
\frac1\pi\int_{-\infty}^{\infty}\frac{\diff}{\diff y}\arg p(iy)\,\diff y=\tau_s-\tau_u,
\]
where $\int_{-\infty}^{\infty}\frac{\diff}{\diff y}\arg f(y)\,\diff y$ is the variation of the argument of $f:\Re\to\Co$ when $y$ goes from $-\infty$ to $+\infty$ (on the real line), and $\tau_s$ and $\tau_u$ are as in Lemma \ref{lem-routh-hurwitz-regular}.
\end{lemma}

\begin{proof}
First, consider a degree-one polynomial $p:x\mapsto(x-\alpha-\beta i)$ with $\alpha\in\Re_{\neq0}$ and $\beta\in\Re$.
Then, the argument of $p(iy)$ is equal to $\arctan((\beta-y)/\alpha)+\nu\pi$ ($\nu\in\ZZb$).
Hence, for a general polynomial $p:x\mapsto(x-\lambda_1)\cdots(x-\lambda_d)$, with $\lambda_s=\alpha_s+\beta_s i$ for each $s\in\{1,\ldots,d\}$, the argument of $p(iy)$ is equal to $\sum_{s=1}^d\arctan((\beta_s-y)/\alpha_s)+\nu\pi$ ($\nu\in\ZZb$).
We readily check (by differentiating and integrating) that each term with $\alpha_s<0$ contributes to a variation of the argument equal to $\pi$ and each term with $\alpha_s>0$ contributes to a variation of the argument equal to $-\pi$.
\end{proof}

% \begin{figure}
% \centering
% \includegraphics[width=\linewidth]{fig_argument}
% \caption{Variation of the argument of two polynomials with multiple roots and of their products.}
% \label{fig-variation-argument}
% \end{figure}

Using the above lemmas, we now conclude the proof of Lemma \ref{lem-routh-hurwitz-regular}.

Note that $p$ has no root on the imaginary axis since $\pred_0$ and $\pred_1$ have no root in common (they have $1$ as GCD).
We use Lemma \ref{lem-variation-argument-roots} and look at the variation of the argument of $p(iy)$ when $y$ varies from $-\infty$ to $+\infty$.
Since $\deg\,\pred_1<\deg\,\pred_0$, the argument of $p(iy)$ converges to a multiple of $\pi$ when $y\to\pm\infty$.
Hence, the variation of the argument, divided by $\pi$, is equal to the number of times $p(iy)$ crosses the imaginary axis in \emph{positive sense} (see Figure \ref{fig-sense-crossing}) minus the number of times it crosses it in \emph{negative sense} (see Figure \ref{fig-sense-crossing}), when $y$ varies from $-\infty$ to $+\infty$.

The first situation (positive crossing) occurs for values of $y$ such that $\pred_0(y)=0$ and $\pred_0\pred_1$ goes from positive to negative at $y$ (note that $\pred_1(y)\neq0$ since $\pred_0$ and $\pred_1$ have no root in common).
If $\pred_0\pred_1$ goes from positive to negative, then this adds one variation of sign at the beginning of the sequence $\pred_0(y),\ldots,\pred_m(y)$: i.e., $\Vred(y^+)=\Vred(y^-)+1$.

The second situation (negative crossing) occurs for values of $y$ such that $\pred_0(y)=0$ and $\pred_0\pred_1$ goes from negative to positive at $y$ (again, $\pred_1(y)\neq0$ since $\pred_0$ and $\pred_1$ have no root in common).
This subtracts one variation of sign at the beginning of $\pred_0(y),\ldots,\pred_m(y)$: i.e., $\Vred(z^+)=\Vred(z^-)-1$.

Now, if $\pred_k(y)=0$ for some $k\in\{1,\ldots,m\}$ and some $y\in\Re$, then $\pred_{k-1}(y)$ and $\pred_{k+1}(y)$ are nonzero (since $\pred_{k\pm1}$ and $\pred_k$ have no root in common) and have opposite signs (since $p_{k+1}(y)=q_k(y)p_k(y)-p_{k-1}(y)=-p_{k-1}(y)$), so that a change of sign of $\pred_k$ at $y$ would not affect the number of sign variations in $\pred_0(y),\ldots,\pred_m(y)$: i.e., $\Vred(y^+)=\Vred(y^-)$.

Putting things together: if $y$ varies from $-\infty$ to $+\infty$, we get that $\Vred(+\infty)-\Vred(-\infty)$ is equal to the variation of the argument of $p(iy)$ when $y$ goes from $-\infty$ to $+\infty$.
Thus, we obtain the conclusion using Lemmas \ref{lem-variation-argument-roots} and \ref{lem-variations-reduced}.\hspace*{\fill}~\QED

\begin{figure}
\centering
\includegraphics[width=0.5\linewidth]{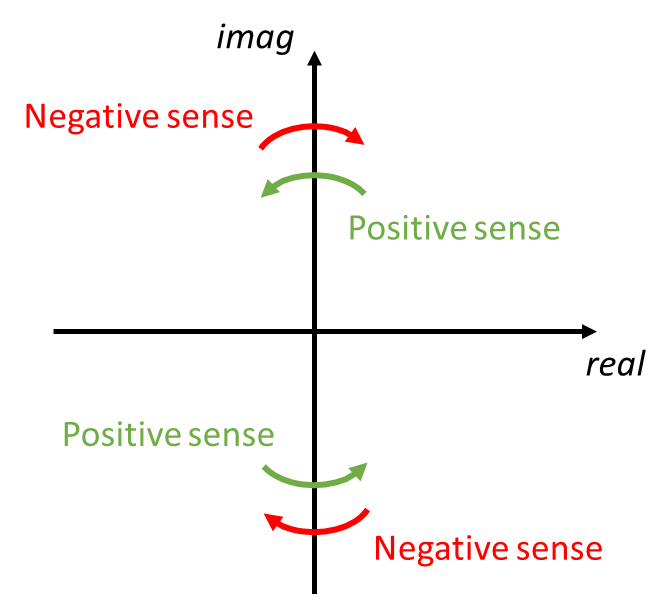}
\caption{Positive and negative sense for crossing the imaginary axis.}
\label{fig-sense-crossing}
\end{figure}

%%%%%%%%%%%%%%%%%%%%%%%%%%%%%%%%%%%%%%%%%%%%%%%%%%%%%%%%%%%%%%%%%%%%%%%%%%%%%%%%
\subsection{Proof of Lemma \ref{lem-sturm-real-roots}}\label{ssec-lem-sturm-real-roots-proof}

We use Lemma \ref{lem-variations-reduced} and look at the variations of sign in the sequence $\pred_0(y),\ldots,\pred_m(y)$.

The hypothesis that $p_0'=-p_1^{}$ implies that if $\lambda$ is a root of $p_0$ with multiplicity $m>1$, then $\lambda$ is a root of $p_1$ with multiplicity $m-1$.
Thus, since $p_m$ is a GCD of $p_0$ and $p_1$ (Lemma \ref{lem-gcd-euclidean-algorithm}), $\lambda$ is also a root of $p_m$ with multiplicity $m-1$.
Hence, the roots of $\pred_0$ are simple and correspond to the distinct roots of $p_0$.

The rest of the proof is similar to the one of Lemma \ref{lem-routh-hurwitz-regular}: we look at the changes in the number of variations of sign in the sequence $\pred_0(y),\allowbreak\ldots,\allowbreak\pred_m(y)$ when $y$ goes from $-\infty$ to $+\infty$.
In particular, the number of sign variations in the sequence changes only when $y$ crosses a real root of $p_0$.
Indeed, in this case, $\pred_0(y)=0$ and $\pred_0^{}\pred_1^{}$ goes from positive to negative at $y$ (since $(\pred_0\pred_1)'(y)<0$).
This adds one variation of sign at the beginning of the sequence $\pred_0(y),\ldots,\pred_m(y)$: i.e., $\Vred(y^+)=\Vred(y^-)+1$.
For other values of $y\in\Re$ such that $\pred_k(y)=0$ for some $k\in\{1,\ldots,m\}$, it holds that $\pred_{k-1}(y)$ and $\pred_{k+1}(y)$ are nonzero and have opposite signs (same as in the proof of Lemma \ref{lem-routh-hurwitz-regular}), so that it does not affect the number of sign variations in $\pred_0(y),\ldots,\pred_m(y)$: i.e., $\Vred(y^+)=\Vred(y^-)$.

Putting things together: we get that $\Vred(+\infty)-\Vred(-\infty)$ is equal to the number of roots of $\pred_0$ on the real line, i.e., to the number of distinct real roots of $p_0$.\hspace*{\fill}~\QED

%%%%%%%%%%%%%%%%%%%%%%%%%%%%%%%%%%%%%%%%%%%%%%%%%%%%%%%%%%%%%%%%%%%%%%%%%%%%%%%%
\subsection{Proof of Lemma \ref{lem-routh-hurwitz-boundedness}}\label{ssec-lem-routh-hurwitz-boundedness-proof}

Let $P=p$ if $\deg\,p$ is even and $P=ip$ if $\deg\,p$ is odd.
In both cases, $P$ and $p$ have the same roots.
By definition of $p_0$ and $p_1$, it holds that $P(x)=p_0(-ix)+ip_1(-ix)$ for all $x\in\Co$.
Thus, letting $\pred_0=p_0/p_m$ and $\pred_1=p_1/p_m$, and since $p_m=\pext_0$, it also holds that
\[
P(x)=(\pred_0^{}(-ix)+i\pred_1^{}(-ix))\pext_0(-ix)
\]
for all $x\in\Co$.
Hence, the roots of $P(x)$ consist of the roots of $P_1:x\mapsto\pred_0(-ix)+i\pred_1(-ix)$ and of $P_2:x\mapsto\pext_0(-ix)$.

Let us look at the first polynomial $P_1$.
By Lemma \ref{lem-routh-hurwitz-regular}, we get that $P_1$ has no root on the imaginary axis, and that the number of its roots with negative real part minus the number of its roots with positive real part is equal to $V(+\infty)-V(-\infty)$.
Thus, $P_1$ has only roots with negative real part if and only if $V(+\infty)-V(-\infty)=\deg\,\pred_0=\deg\,p_0-\deg\,\pext_0$, otherwise $P_1$ (and thus $P$) has at least one root with positive real part.

Next, let us look at the polynomial $P_2$.
From Lemma \ref{lem-sturm-real-roots}, we have that $\pext_0$ has only simple real roots if and only if $\Vext(+\infty)-\Vext(-\infty)=\deg\,\pext_0$, otherwise $\pext_0$ has a at least one multiple real root or a (multiple or not) complex root.
Note that since $\pext_0$ is real, any complex root $\alpha+\beta i$ of $\pext_0$ gives rise to another root $\alpha-\beta i$ (complex conjugate).
Thus, the conclusions for $P_2$ are the following: $P_2$ has only simple roots that are all on the imaginary axis if and only if $\Vext(+\infty)-\Vext(-\infty)=\deg\,\pext_0$.
Otherwise $P_2$ (and thus $P$) has either at least one multiple root on the imaginary axis or at least one root with positive real part.

Hence, the roots of $P$ (and thus those of $p$ too) satisfy (i) or (ii) in Theorem \ref{thm-roots-minimal-polynomial-boundedness-cont} if and only if $V(+\infty)-V(-\infty)=\deg\,p_0-\deg\,\pext_0$ and $\Vext(+\infty)\allowbreak-\Vext(-\infty)=\deg\,\pext_0$, concluding the proof of the lemma.\hspace*{\fill}~\QED

%%%%%%%%%%%%%%%%%%%%%%%%%%%%%%%%%%%%%%%%%%%%%%%%%%%%%%%%%%%%%%%%%%%%%%%%%%%%%%%%
\subsection{Proof of Corollary \ref{cor-routh-hurwitz-boundedness-degree-dec}}\label{ssec-cor-routh-hurwitz-boundedness-degree-dec-proof}

If $\deg\,p_k<\deg\,p_{k-1}-1$ for some $k\in\{1,\allowbreak\ldots,\allowbreak m\}$, then $m<\deg\,p_0-\deg\,p_m$, so that $V(+\infty)-V(-\infty)<\deg\,p_0\allowbreak-\deg\,\pext_0$.
Similarly, if $\deg\,\pext_k<\deg\,\pext_{k-1}-1$ for some $k\in\{1,\allowbreak\ldots,\allowbreak\mext\}$, then $\mext<\deg\,\pext_0$, so that $\Vext(+\infty)\allowbreak-\Vext(-\infty)<\deg\,\pext_0$.
Thus, if one of the two situations occurs, then \eqref{eq-sum-variations} does not hold, concluding the proof.\hspace*{\fill}~\QED

%%%%%%%%%%%%%%%%%%%%%%%%%%%%%%%%%%%%%%%%%%%%%%%%%%%%%%%%%%%%%%%%%%%%%%%%%%%%%%%%
\subsection{Proof of Theorem \ref{thm-roots-minimal-polynomial-boundedness-disc}}\label{ssec-thm-roots-minimal-polynomial-boundedness-disc-proof}

Similar to the proof of Theorem \ref{thm-roots-minimal-polynomial-boundedness-cont}.
The only difference is the expression of the \emph{power} of a Jordan block $J_n(\lambda)$: $(J_n(\lambda))^t=\sum_{k=0}^{n-1}\binom{t}{k}\lambda^{t-k}(J_n(0))^k$, which grows faster that $\lambda^t$ if $n>1$ (via the binomial coefficient $\binom{t}{k}$).\hspace*{\fill}~\QED

%%%%%%%%%%%%%%%%%%%%%%%%%%%%%%%%%%%%%%%%%%%%%%%%%%%%%%%%%%%%%%%%%%%%%%%%%%%%%%%%
\subsection{Proof of Theorem \ref{thm-analyze-roots-polynomial-disc}}\label{ssec-thm-analyze-roots-polynomial-disc-proof}

The computation of the polynomial $p$ relies on the function $f:x\mapsto\frac{x+1}{x-1}$, called a \emph{M\"obius trans\-formation} and which has the following properties:

\underline{\smash{Properties of $f$:}}
The function $f$ defined above is bijective between $\Co_{\neq1}$ and $\Co_{\neq1}$, and is its own inverse (i.e., $f\circ f=\mathrm{id}$).
Moreover, for any $x\in\Co_{\neq1}$, it holds that $x$ has negative real part if and only if $\lvert f(x)\rvert<1$; and $x$ is on the imaginary axis if and only if $\lvert f(x)\rvert=1$.\hfill~$\triangleleft$

Define the polynomial $P:x\mapsto\sum_{\ell=0}^d a_\ell(x+1)^{d-\ell}(x-1)^\ell$, which can be computed in polynomial time by expanding the factors $(x\pm1)^\ell$ ($\ell\in\{0,\ldots,d\}$) and re\-arranging the terms (see also the discussion on the bit size of $A^\ell$ in the proof of Proposition \ref{thm-compute-minimal-polynomial}).
It is readily checked that $P$ coincides with $x\mapsto\pdisc(\frac{x+1}{x-1})(x-1)^d$ on $\Co_{\neq1}$.

Hence, if $\pdisc(x)=a_0(x-1)^{n_0}(x-\lambda_1)^{n_1}\cdots(x-\lambda_m)^{n_m}$ for all $x\in\Co$ where $n_0\in\ZZb_{\geq0}$ and $n_1,\ldots,n_m\in\ZZb_{>0}$, then it holds that $P(x)=a_02^{n_0}\prod_{s=1}^m(x+1-\lambda_s(x-1))^{n_s}$ for all $x\in\Co$.

Hence, the function $f$ is a one-to-one mapping between the roots (\emph{with multiplicity}) of $\pdisc$ in $\Co_{\neq1}$ and the roots (\emph{with multiplicity}) of $P$ in $\Co_{\neq1}$.
Moreover, the ``degree difference'' $\delta\coloneqq\deg\,\pdisc-\deg\,P$ corresponds to the multiplicity ($n_0$) of $1$ as a root of $\pdisc$.
Hence, letting $p:x\mapsto (x-1)^\delta P(x)$, we get the desired output of the algorithm.\hspace*{\fill}~\QED

% Consider the function $g\coloneqq\pdisc\,\circ f:\Co_{\neq1}\to\Co:$
% \begin{equation}\label{eq-polynomial-mobius}
% x\mapsto \frac{\sum_{\ell=0}^d a_\ell(x+1)^{d-\ell}(x-1)^\ell}{(x-1)^d}.
% \end{equation}
% Let $P$ be the polynomial equal to the numerator of \eqref{eq-polynomial-mobius}.
% N

% Let $\pdisc:x\mapsto a_0(x-1)^{n_0}(x-\lambda_1)^{n_1}\cdots(x-\lambda_m)^{n_m}$ where $n_0\in\ZZb_{\geq0}$ and $n_1,\ldots,n_m\in\ZZb_{>0}$.
% Then, it holds that $g:x\mapsto a_02^{n_0}\prod_{s=1}^m(x+1-\lambda_s(x-1))^{n_s}/(x-1)^d$.
% So $P:x\mapsto a_02^{n_0}\prod_{s=1}^m(x+1-\lambda_s(x-1))^{n_s}$.
% Hence, the function $f$ is a one-to-one mapping between the roots (\emph{with multiplicity}) of $\pdisc$ in $\Co_{\neq1}$ and the roots (\emph{with multiplicity}) of $P$ in $\Co_{\neq1}$.
% Moreover, the ``degree difference'' $\delta\coloneqq\deg\,\pdisc-\deg\,P$ corresponds to the multiplicity ($n_0$) of $1$ as a root of $\pdisc$.
% Hence, letting $p:x\mapsto (x-1)^\delta P(x)$, we get the desired output of the algorithm; concluding the proof.

}

\end{document}